\documentclass{article}

\usepackage{graphicx}
\usepackage{subfig}
\usepackage{wrapfig}
\usepackage{url}
\usepackage{amsfonts}
\usepackage{amsmath}
\usepackage{amssymb}
\usepackage[mathscr]{euscript}
\usepackage{stmaryrd}

\usepackage{paralist}
\usepackage{array}
\usepackage{color}
\usepackage{hyperref}

\usepackage{pgf}

\usepackage[all]{xy}
\usepackage{amsthm}

\usepackage{colortbl}

\usepackage{algorithm}
\usepackage{algorithmic}

\usepackage{todonotes}
\usepackage{savesym}
\savesymbol{vec}
\usepackage{MnSymbol}
\restoresymbol{Mn}{vec}

\graphicspath{{fig/}}

\makeatletter
\newcommand\nbh{\hbox{-}\nobreak\hskip\z@skip}
\makeatother

\DeclareMathAlphabet{\mathpzc}{OT1}{pzc}{m}{it}
\usepackage{mathbbol}

\usepackage[affil-it]{authblk}

\newcommand{\pre}[1]{\bullet{#1}}

\newcommand{\reportonly}[1]{}

\newcommand{\ac}{\ensuremath{\nearrow}}
\newcommand{\scauses}[1]{\lfloor #1 )}
\newcommand{\causes}[1]{\lfloor #1 \rfloor}

\newcommand{\hist}[1]{\mathit{hist}({#1})}
\newcommand{\allHistories}[1]{\check{\mathcal{H}}(#1)}
\newcommand{\conf}[1]{\ensuremath{\mathit{Conf}({#1})}}
\newcommand{\fold}[2]{\ensuremath{{#1}_{/{#2}}}}
\newcommand{\mcons}[1]{\ensuremath{\mathbb{C}({#1})}}

\newcommand{\dConflict}{\#_\mu}

\makeatletter
\newtheorem*{rep@theorem}{\rep@title} \newcommand{\newreptheorem}[2]{%
\newenvironment{rep#1}[1]{%
\def\rep@title{\bf #2 \ref{##1}}%
\begin{rep@theorem} }%
{\end{rep@theorem} } }
\makeatother
\newreptheorem{theorem}{Theorem}
\newreptheorem{lemma}{Lemma}

\theoremstyle{definition}
\newtheorem{definition}{Definition}
\theoremstyle{plain}
\newtheorem{theorem}{Theorem} 
\newtheorem{lemma}[theorem]{Lemma}
\newtheorem{corollary}[theorem]{Corollary}

\usepackage{thmtools, thm-restate}
\newif\ifTechnicalReport
\TechnicalReportfalse

\pgfdeclarelayer{background}
\pgfdeclarelayer{foreground}
\pgfsetlayers{background,main,foreground}

\begin{document}


	\title{Reduction of Event Structures under History Preserving Bisimulation}

	\author{
		Abel Armas-Cervantes%
			\thanks{\texttt{abel.armas@ut.ee}}}
	\affil{ Institute of Computer Science, University of Tartu, Estonia.}
	 \author{
		Paolo Baldan%
		 \thanks{\texttt{baldan@math.unipd.it}}}
	\affil{Department of Mathematics, University of Padova, Italy.}
	 \author{
		Luciano Garc\'ia-Ba\~nuelos%
		\thanks{\texttt{luciano.garcia@ut.ee}}}
	\affil{ Institute of Computer Science, University of Tartu, Estonia.}
	
\date{}

\maketitle

\begin{abstract}
  Event structures represent concurrent processes in terms of events
  and dependencies between events modelling behavioural relations like
  causality and conflict. Since the introduction of prime event
  structures, many variants of event structures have been proposed
  with different behavioural relations and, hence,  with differences in their
  expressive power. 
  One of the possible benefits of using a more expressive event structure is
  that of having a more compact representation for the same behaviour when 
  considering the number of events used in a prime event structure. 
  Therefore, this article addresses the problem of reducing the size of an
  event structure while preserving behaviour under a well-known
  notion of equivalence, namely history preserving bisimulation. In particular, 
  we investigate this problem 
  on two generalisations of the prime event structures.
  The first one, known as asymmetric event structure, relies on a asymmetric
  form of the conflict relation. The second one, known as flow event structure,
  supports a form of disjunctive causality. 
  More specifically, we describe the conditions under which a set of
  events in an event structure can be folded into a single event while
  preserving the original behaviour. 
  The successive application of this folding operation leads to a
  minimal size event structure. However, the order on which the
  folding operation is applied may lead to different minimal size 
  event structures. The latter has a negative implication on the
  potential use of a minimal size event structure as a canonical 
  representation for behaviour.
 
\end{abstract}



\section{Introduction}

The concept of concurrent process is pervasive in computer science,
with applications in a multitude of distinct fields, and a wide range
of formalisms and techniques have been developed for the modelling and
analysis of processes. 
%
%
%
Event structures are one of the possible formalisms for modelling concurrent processes.
Computations underlying the execution of processes are represented by means of events and behavioural
relations. Events represent occurrences of atomic
actions. Behavioural relations, which differ in the various types of event
structures, explain how events relate each other. For instance, when the
occurrence of one event requires another event to occur beforehand, 
we say that there is a causal relation between them. Similarly,
when the occurrence of one event prevents the occurrence of another event,
we say that they are in conflict relation.
%
In this context, the seminal
work~\cite{Winskel87,NielsenPW81} introduces \emph{prime event
  structures} (PESs), where dependencies between events are reduced to
causality and conflict. Since then, many different types of event
structures have been proposed.
In this work, we consider two other types of event structures, namely,
the \emph{flow event structures} (FESs)~\cite{BoudolC88} 
and \emph{asymmetric event structures} (AESs)~\cite{BaldanCM01}, which
provide a form of disjunctive causality and an asymmetric version
of conflict, respectively.



\begin{figure}[ht]
\centering
\subfloat[PES]{\label{fig:pes0a}
	\includegraphics{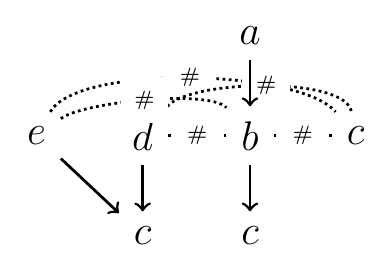}
}
\subfloat[AES]{\label{fig:aes0a}
	\includegraphics{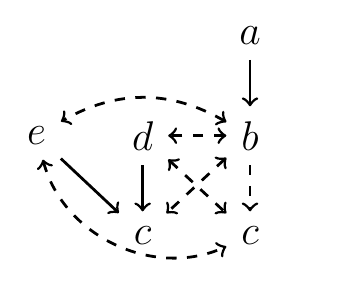}
}
\subfloat[FES]{\label{fig:fes0a}
	\includegraphics{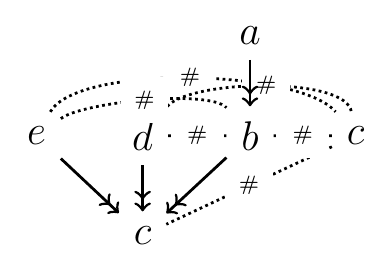}
}
\caption{Three history preserving bisimilar event structures}
\label{fig:pesFes0}
\end{figure}

In order to give a more precise idea of the kind of structures the
paper deals with and of the results we aim at, consider the event
structures depicted in the form of graphs in Fig.~\ref{fig:pesFes0}. 
In all cases, nodes represent events and edges represent behaviour 
relations.
Fig.~\ref{fig:pes0a} presents a PES. There, the straight arrows represent
causality and the annotated dotted edges represent conflict.
%
For instance, events $a$ and $b$ are connected with a straight arrow
and, hence, are in causal relation. Intuitively we say that
\emph{``in a given computation, event $a$
  must occur before $b$''}. Causality in PES is transitive relation. For the sake
  of clarity, only direct causal relations are shown. 
  Similarly, events $d$ and $b$ are
in conflict, which can be intuitively stated as 
\emph{``in a given computation, either $d$ or $b$ occurs but not both''}.

Figure~\ref{fig:aes0a} presents an AES. There, events are related via causality,
which is depicted again with straight arrows, and an asymmetric form of the causal relation, which is
depicted with dashed arrows. The asymmetric conflict relation has two intuitive interpretations.
The fact that $b$ is in asymmetric conflict with
$c$ can be interpreted as \emph{``the occurrence of event $c$ avoids the
occurrence of event $b$''}. The same fact can be interpreted as
\emph{``whenever $b$ and $c$ occur in a computation, then the
execution of event $b$ will precede the execution of $c$''}. For this reason,
the asymmetric conflict can be seen as a weak form of causality. Moreover, 
the symmetric conflict relation can be expressed with asymmetric conflict in
both directions and, more generally speaking,
by means of cycles of them.


Finally, Figure~\ref{fig:fes0a} provides an example of a FES. There, causality is 
replaced by the flow relation, which is represented with a double-headed arrow.
The flow relation is intransitive. Intuitively, the flow relation expresses the set of
potential direct causes for a given event. That, in order for an event to occur a
maximal, conflict set of direct predecessors has to occur beforehand.
For instance, in the
example, the leftmost event with label $c$ must be preceded either by
$\{ e, d \}$ or $\{ b \}$.

Interestingly, the three event structures depicted in
Figures~\ref{fig:pesFes0}(a)-(c) represent the same set of
computations, with a different number of events. This is only possible
because of the greater expressiveness of AESs and FESs. The result
is that the same behaviour is represented with less events in both
cases.
Also, it should be noted that any PES can be straightforwardly transformed
into a AES or into a FES. For the case of AESs, the conflict relation is
translated into two asymmetric conflict arrows and, for the case of FESs,
the flow relation corresponds to the transitive reduction of causality.

The purpose of this article is to introduce transformations for 
reducing the size of AESs and FESs.
Intuitively,
the method requires identifying sets of events that can replaced by a single
event, while preserving the original behaviour.
The method entails
a morphism on event structures, referred to as \emph{folding}, that
is shown to preserve a well-known notion of equivalence,
namely history preserving bisimulation~\cite{RavinovitchT88,vanGlabbeekG89,BestDKP91}.
For instance, the AES and FES presented in Fig.~\ref{fig:aes0a} and~\ref{fig:fes0a},
respectively, can be obtained by folding subsets of occurrences of the event
$c$ that are present in the PES shown in Fig.~\ref{fig:pes0a}.
This notion of equivalence is one of classical behavioural equivalences 
in the true concurrency spectrum.  The iterative folding of a finite
event structure eventually converge into a (locally) minimal event
structure.
%
Unfortunately, the minimal event structure is not always unique and, therefore,
cannot be used as a canonical representation.



The organisation of the paper is as follows,
Section~\ref{sec:preliminaries} introduces basic concepts about the
notation, the event structures used and the adopted equivalence. The
folding technique over AES is presented in
Section~\ref{sec:reductionAES}.  The folding technique defined for FES
is presented in
Finally,
Section~\ref{sec:conclusionFutureWork} draws some conclusions and
proposes possible avenues for future work.

\section[Prime event structures and hp-bisimilarity]{Prime event structures and history preserving bisimilarity}
\label{sec:preliminaries}

This section recalls the basics of \emph{prime event structures} and
introduces the notion of \emph{history-preserving bisimilarity}, that
will provide a foundation to discussion in the following sections.

We shall first recall some basic notation on sets and relations. Let
$R \subseteq X \times X$ be a binary relation and let $Y \subseteq X$,
then $R|_Y$ denotes the restriction of $R$ to $Y$, i.e., $R|_Y = R
\cap (Y \times Y)$. We say that $R$ is \emph{well-founded} if it has
no infinite descending chain, i.e., $\langle e_i \rangle_{i \in
  \mathbb{N}}$ such that $e_{i+1}\ R\ e_i,\ e_i \neq e_{i+1}$, for all
$i \in \mathbb{N}$. The relation $R$ is \emph{acyclic} if it has no
``cycles'', that is, $e_0\ R\ e_1\ R \dots R\ e_n\ R\ e_0$ with $e_i \in X$, does not hold. In
particular, if $R$ is well-founded, then it has no (non-trivial)
cycles. Relation $R$ is a \emph{preorder}, if it is reflexive and
transitive; it is a \emph{partial order} if it is also antisymmetric.

  

\subsection{Prime Event structures}

We recall the formal definition of \emph{prime event
  structures}~\cite{NielsenPW81} which complements the informal
description provided in the introduction. Hereafter $\Lambda$ denotes
a fixed set of labels.

\begin{definition}[prime event structure]
  \label{def:pes}
  A (labelled) \emph{prime event structure} (PES) is a tuple
  $\mathbb{P} = \langle E, \leq, \#, \lambda \rangle$, where  $E$
  is a set of events, $\leq$ and $\#$ are binary relations on $E$ called \emph{causality} and \emph{conflict}, and $\lambda : E \to \Lambda$ is a labelling function, such that 
  \begin{compactitem}

  \item $\leq$ is a partial order and $\causes{e} = \{e' \in E \mid e'
    \leq e\}$ is finite for all $e \in E$;

  \item $\#$ is irreflexive, symmetric and hereditary with respect to
    causality, i.e., for all $e,e',e'' \in E$, if $e\#e' \leq e''$
    then $e \# e''$
  \end{compactitem}
\end{definition}

An event $e \in E$ labelled with $a$ represents the occurrence of an
action $a$ in a computation of the system, $e < e'$ means that $e$ is a
prerequisite for the occurrence of $e'$ and $e \# e'$ means that $e$
and $e'$ cannot both happen in the same computation. In order to lighten the
notation, whenever it is clear from the context, we will use
events and event labels interchangeably.

The computations in an event structure are usually described in terms
of configurations, i.e., sets of events which are closed with respect
to causality and conflict free. Formally, a \emph{configuration} of a
PES $\mathbb{P} = \langle E, \leq, \#,\lambda \rangle$ is a finite set
of events $C \subseteq E$ such that
\begin{compactitem}
\item for all $e\in C$, $\causes{e}
\subseteq C$ and
\item for all $e,e' \in C$, $\lnot(e \# e')$.
\end{compactitem}
Configurations come equipped with an extension order, $C_1 \sqsubseteq
C_2$ meaning that a configuration $C_1$ can evolve into $C_2$. For
PESs, the extension order is simply subset inclusion.

\subsection{History preserving bisimilarity}

In this paper we use the notion of history
preserving
bisimilarity~\cite{RavinovitchT88,vanGlabbeekG89,BestDKP91}, a
classical equivalence in the true-concurrency spectrum. As for
bisimilarity in interleaving semantics, an event of an event structure must be
simulated by an event of the other, with the same label, and
vice-versa, but additionally, the two events are required to have the
same ``causal history''.

\begin{definition}[history preserving bisimilarity]
  Let $\mathbb{E}_1$, $\mathbb{E}_2$ be two PESs. A \emph{history
    preserving (hp-)bisimulation} is a set $R$ of triples $(C_1, f,
  C_2)$, where $C_1$ and $C_2$ are configurations of $\mathbb{E}_1$
  and $\mathbb{E}_2$, respectively, and $f$ is an isomorphism, such
  that $(\emptyset,\emptyset,\emptyset) \in R$ and 
  $\forall (C_1, f,  C_2)\in R$
    \begin{enumerate}[a)]
    \item if $C_1 \cup \{e_1\} \in \conf{\mathbb{E}_1}$, for an event
      $e_1 \in \mathbb{E}_1$, there exists $e_2 \in \mathbb{E}_2$ such
      that $\lambda_1(e_1) = \lambda_2(e_2)$ and $(C_1 \cup \{e_1\},
      f', C_2 \cup \{e_2\}) \in R$;
      
    \item if $C_2 \cup \{e_2\} \in \conf{\mathbb{E}_2}$, for an event
      $e_2 \in \mathbb{E}_2$, there exists $e_1 \in \mathbb{E}_1$ such
      that $\lambda_1(e_1) = \lambda_2(e_2)$ and $(C_1 \cup \{e_1\},
      f', C_2 \cup \{e_2\}) \in R$.
    \end{enumerate}
    Moreover $\mathbb{E}_1$, $\mathbb{E}_2$ are said history preserving
    bisimulation equivalent or, simply, history preserving bisimular iff
    the bisimulation R exists.
\end{definition}

Although hp-bisimilarty is defined only for PESs, the same notion can be 
straightforwardly adapted to the other variants of event structures used in 
this article.

\section{Behaviour-Preserving Reduction of AESs}
\label{sec:reductionAES}

In this section we describe a technique for reducing the size of
asymmetric event structures, in a way that preserves their behaviour.

\subsection{Basics of asymmetric event structures}

We briefly review the basics of asymmetric event structures. 

\begin{definition}[asymmetric event structure]
  A (labelled) \emph{asymmetric event structure} (AES) is a tuple
  $\mathbb{A}= \langle E, \leq, \ac, \lambda \rangle$, where $E$ is a
  set of events, $\leq$ and $\ac$ are binary relations on $E$ called
  \emph{causality} and \emph{asymmetric conflict}, and $\lambda : E
  \to \Lambda$ is a labelling function, such that
  \begin{compactitem}
  \item $\leq$ is a partial order and $\causes{e} = \{e' \in E \mid e'
    \leq e\}$ is finite for all $e \in E$;

  \item $\ac$ satisfies, for all $e, e', e'' \in E$
    \begin{compactenum}
    \item $e < e' \Rightarrow e\ac e'$,
    \item if $e \ac e'$ and $e' < e''$ then $e \ac e''$;
    \item $\ac|_{\causes{e}}$ is acyclic;
    \item if $\ac|_{\causes{e} \cup \causes{e'}}$ is cyclic then $e \ac e'$.
    \end{compactenum}
  \end{compactitem} 
  \label{def:AES}
\end{definition}

AESs generalise PESs by allowing a conflict relation which is no
longer symmetric. As hinted at in the introduction, the asymmetric
conflict relation has a double interpretation, that is $a \ac b$ can
be understood as
\begin{inparaenum}[(i)]
  \item \label{int:ac:1} the occurrence of $a$ \emph{is prevented by} $b$, or
  \item \label{int:ac:2} $a$ \emph{precedes} $b$ in all computations where both 
    appear.
\end{inparaenum}
Condition 1 of Definition~\ref{def:AES} arises from the fact that,
according to the interpretation (\ref{int:ac:2}) of the asymmetry
conflict relation, $\ac$ can be seen as a weak form of causality,
hence it is natural to ask that it is included in $<$. In the
graphical representation of an AES, $\leq$ takes precedence over $\ac$
and, therefore, when both holds a solid edge is used.
Condition 2 is a form of hereditarity of asymmetric conflict along
causality: if $e \ac e'$ and $e' < e''$ then $e$ is necessarily
executed before $e''$ when both appear in the same computation, hence
$e \ac e''$ (see Fig.~\ref{fig:heredityAC}(a)).  Concerning conditions
3 and 4, observe that events forming a cycle of asymmetric conflict
cannot appear in the same run, since each event in the cycle should
occur before itself in the run. This leads to a notion of
\emph{conflict} over sets of events $\# X$, defined by the following
rules
\[
\frac{e_0 \ac e_1 \ac \ldots \ac e_n \ac e_0}{\# \{ e_0, \ldots, e_n \}}
\qquad \qquad
\frac{ \#( X \cup \{ e \}) \ e \leq e'}{ \#( X \cup \{ e' \})}
\]
In this view, condition 3 corresponds to irreflexiveness of conflict
in PES, while condition 4 requires that binary symmetric conflict are
represented by asymmetric conflict in both directions.
 

In the following, direct relations, namely immediate causality and
conflicts that are not inherited, will play a special role.

\begin{definition}[direct relations]
  Let $\mathbb{A}$ be an AES and let $e, e' \in E$. We say that $e'$
  is an \emph{immediate cause} of $e$, denoted $e' <_\mu e$, when $e'
  < e$ and there is no $e''$ such that $e' < e'' < e$.  An asymmetric
  conflict $e \ac e''$ is called \emph{direct}, written $e \ac_\mu
  e''$ when there is no $e'$ such that $e \ac e' < e''$. A binary
  conflict $e \# e'$ is called \emph{direct}, written $e \#_\mu e'$,
  when $e \ac_\mu e'$ and $e' \ac_\mu e$.
\end{definition}

\begin{figure}
\centering
\subfloat[$e \ac_\mu e'$ and $e \ac e' $]{
\makebox[.4\textwidth]{
   \includegraphics{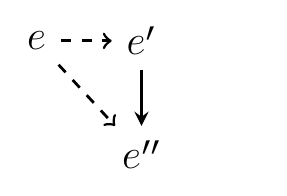}
}
}
\subfloat[$e' \ac_\mu e$ and $e \dConflict e'' $]{
 \makebox[.4\textwidth]{
  \includegraphics{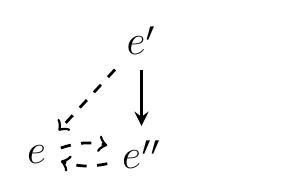}
  }
}
\caption{Hereditarity of $\ac$}
\label{fig:heredityAC}
\end{figure}
For instance, in Fig.~\ref{fig:heredityAC}(a) $e \ac_\mu e'$ while it
is not the case that $e \ac_\mu e''$. In Fig.~\ref{fig:heredityAC}(b)
we have that $e'' \ac_\mu e$ and $e \ac_\mu e''$, hence $e \dConflict
e''$.


Configurations in AES are defined, as in PES, as causally closed and
conflict free set of events. More precisely a configuration of
$\mathbb{A} = \langle E, \leq, \ac,\lambda \rangle$ is a set of events
$C \subseteq E$ such that
\begin{inparaenum}[1)]
\item for any $e \in C$, $\causes{e} \subseteq C$ (causal closedness)
\item $\ac|_C$ is acyclic (or equivalently, $\lnot (e \# e')$ for all
  $e, e' \in C$).
\end{inparaenum}
The set of all configurations of $\mathbb{A}$ is denoted by $\conf{\mathbb{A}}$. 


Differently from what happens for PES, the extension order on
configurations is not simply set-inclusion, since a configuration $C$
cannot be extended with an event which is prevented by some of the
events already present in $C$. More formally, if $C_1, C_2 \in
\conf{\mathbb{A}}$ are configurations, we say that $C_2$ extends $C_1$,
written $C_1 \sqsubseteq C_2$, if $C_1 \subseteq C_2$ and for all
$e\in C_1, \ e' \in C_2 \setminus C_1$, $\lnot(e' \ac e)$.

A fundamental notion is that of history of an event in a configuration.

\begin{definition}[history and possible histories]
  \label{def:possibleHist}
  Let $\mathbb{A}$ be an AES and
  let $e \in E$ be an event in $\mathbb{A}$. Given a configuration
  $C\in \conf{\mathbb{A}}$ such that $e\in C$, the \emph{history} of
  $e\in C$ is defined as $C \llbracket e \rrbracket = \{e' \in C \mid
  e' (\ac|_C)^* e\}$. The \emph{set of possible histories} of $e$,
  denoted by $\hist{e}$, is then defined as
  \begin{center}
    $\hist{e} = \{C \llbracket e \rrbracket \mid C \in
    \conf{\mathbb{A}} \land e\in C\}$
  \end{center}
  We will write $\allHistories{e} = \bigcup \hist{e}$ to represent the
  the set of all events possibly occurring in a history of event
  $e$. Moreover, given a history $h \in \hist{e}$, we define $h^- = h
  \setminus \{e\}$.
\end{definition}

Roughly speaking, $C \llbracket e \rrbracket$ consists of the events
which necessarily must occur before $e$ in the configuration
$C$. While in the case of PESs, each event $e$ has a unique history,
i.e., the set $\causes{e}$, in the case of AESs, an event $e$ may have
several histories. For example, the event $c_{0,2}$ in the AES $\mathbb{A}_2$ 
(Figure~\ref{fig:exampleAES}(c)) has four different histories, $\hist{c_{0,2}} 
= \{\{c_{0,2}\}, \{d, c_{0,2}\}, \{e, c_{0,2}\}, \{d, e, c_{0,2}\}\}$.

With abuse of notation, we will use
$\hist{X} = \bigcup_{e \in X} \hist{e}$ to denote the set of events in
the history of a set of events $X$.

\subsection{Reduction of AESs}

The technique for behaviour preserving reduction of AESs consists in
iteratively identifying a set of events carrying the same label, i.e.,
intuitively referring to the same action, and replacing all the events
in the set with a single event. Such a substitution is called a
\emph{folding}.  However, the configurations of the AES should remain
``essentially'' unchanged after the folding or, more precisely, the
original and the folded AES should be hp-bisimilar.

In order to understand the intuition behind folding, consider the
sample AESs in Figure~\ref{fig:exampleAES}, where events are named
using their label, possibly with subscripts (e.g., $c_0$ is an event
labelled by $c$). The AES $\mathbb{A}_1$ can be thought of as a
reduction of $\mathbb{A}_0$ obtained by folding two $c$-labelled events
$c_0$ and $c_1$, the first in conflict with $d$ and the second caused
by $d$, into a single event $c_{0,1}$, in asymmetric conflict with
$d$.  The dependencies $d\ \#\ c_0$ and $d < c_1$ in $\mathbb{A}_0$ give
rise to an asymmetric conflict, i.e., $d \ac c_{0,1}$ in
$\mathbb{A}_1$, as a side effect of the substitution.

\begin{figure}
\centering
\subfloat[\label{fig:AES:original}$\mathbb{A}_0$]{
  \includegraphics{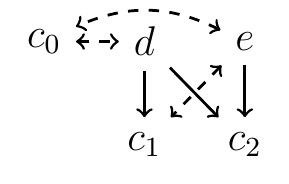}
}
\hspace{10mm}
\subfloat[\label{fig:AES:equivalent}$\mathbb{A}_1$]{
  \includegraphics{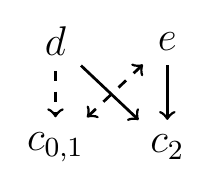}
  }
\hspace{10mm}
\subfloat[\label{fig:AES:nonequivalent}$\mathbb{A}_2$]{
  \includegraphics{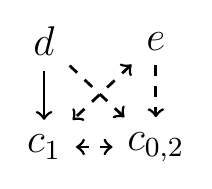}
}
\caption{AESs such that  $\mathbb{A}_0 \equiv_{hp} \mathbb{A}_1$  but $\mathbb{A}_0 \nequiv_{hp} \mathbb{A}_2$.}
\label{fig:exampleAES}
\end{figure}

The configurations of $\mathbb{A}_0$ and $\mathbb{A}_1$, are 
$\conf{\mathbb{A}_0} = \{\{c_0\},\{d, c_1\},\{d, e, c_2\}\}$ and 
$\conf{\mathbb{A}_1} = \{\{c_{0,1}\},\{d, c_{0,1}\},$ $\{d, e, c_2\}\}$, 
and it is not difficult to see that the two AESs are hp-bisimilar.

Also $\mathbb{A}_2$ could look as a reduced version of $\mathbb{A}_0$
where $c_0$ and $c_2$ are folded into $c_{0,2}$. However, this folding
would not preserve the behaviour.  In fact, $\conf{\mathbb{A}_2} =
\{\{c_{0,2}\},\{d, c_1\},\{e, c_{0,2}\}, \{d, c_1\},\{d, e,
c_{0,2}\}\}$ contains an additional configuration not in
$\conf{\mathbb{A}_0}$. This immediately implies that $\mathbb{A}_0$ is
not hp-bisimilar to $\mathbb{A}_2$.

We next identify sets of events that can be safely folded. For this we
need some further notation. Given a set $X$ of events, whenever it can
be folded, the resulting merged event will have as causes the common
causes of all the events in $X$, while events which are causes or weak
causes of only some of the events in $X$ will become weak causes of
the merged event. More precisely, given a set of events $X$, we define
its \emph{strict causes} $S(X) = \bigcap_{x \in X} \scauses{x} = \{ e' \mid \forall e \in X. e' < e \}$ and the
\emph{weak predecessors} as
\begin{center}
  $W(X) = \{ e'' \mid \exists e, e' \in X. \, e'' \ac  e\ \land\ \lnot ( e' \ac e'') \}  \setminus (S(X) \cup X)$
\end{center}

The set $W(X)$ consists of all $\ac$-predecessors of any event $e \in
X$ that is not a strong cause and it is not in conflict with at least
one event in $X$, so that $e$ can appear in the same configuration of
some event in $X$. For instance, in Fig.~\ref{fig:exampleAES}, we have
that $S(\{c_0, c_1\}) = \emptyset$ and $W(\{c_0, c_1\}) = \{ d
\}$. Instead, $S(\{c_1, c_2\}) = \{ d \}$ and $W(\{c_1, c_2\}) = \{ e
\}$. Observe that events in $W(X)$ are not necessarily in a history of
some event in $X$, as shown by the AES in Fig.~\ref{fig:no-hist}.
More specifically, there is a configuration where both events $a$ and 
$c_1$ occur, although $a$ is a weak predecessor because of its $\ac$
relation with $c_0$.

\begin{figure}[h!]
\centering
\includegraphics{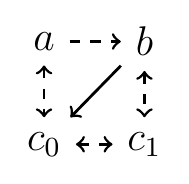}
\caption{The weak predecessors for the set of similar events $X = \{ c_0, c_1\}$ are $W(X) = \{ a, b\}$, but it includes the event $a$ that is not in the history of any event in $X$.}
\label{fig:no-hist}
\end{figure}

A first notion is that of similar events.


\begin{definition}[similar events]
  \label{def:EquivalentEvts}
  Let $\mathbb{A}=\langle E, \leq, \ac, \lambda \rangle$ be an AES.
  A set of events $X \subseteq E$ is called \emph{similar} if for all $e, 
  e' \in X$, $e'' \in E \setminus X$:
  \begin {compactenum}
  \item $\lambda(e) = \lambda(e')$ and $e \# \ e'$
    
  \item $e \ac e'' \quad \Rightarrow \quad e' \ac e''\  \lor\ e'' \ac e$; 
   \item $e'' \ac_\mu e \quad \Rightarrow \quad e'' \ac e'$.
 \end{compactenum}
\end{definition}

    

Intuitively, events to be folded should represent different occurrences 
of the same activity, hence the first condition is that they need to have 
the same label and be in conflict.
Conditions 2 and 3 roughly ask that all events in $X$ have,
essentially, the same asymmetric conflicts (with the exception of
those involving events in the histories).
More precisely, given two events $e, e' \in X$, if for an event $e''$
we have $e \ac e''$ conditions 2 requires that also $e' \ac e''$
or $e'' \ac e$ (and thus $e \# e''$) as it could happen, e.g., when $e''$ is  part of some history of $e'$ but not of $e$.  This can
be understood as follows: we would like to see $e$ and $e'$ as
occurrences of the same activity with different histories, hence $e''$
plays the role of a weak cause, inserted in the history of $e'$ and
incompatible with the history of $e$.
Finally, condition 3 requires that direct $\ac$-predecessors are
preserved in $X$.

\begin{wrapfigure}{r}{0.4\textwidth}
\centering
\subfloat[\label{fig:AES:rule2}$\mathbb{A}_3$]{
   \includegraphics{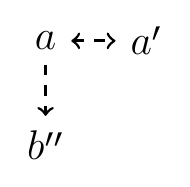}
}
\subfloat[\label{fig:AES:rule3}$\mathbb{A}_3'$]{
  \includegraphics{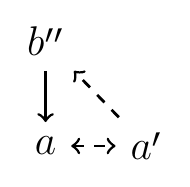}
}
\caption{Examples of non-similar events \emph{a} and \emph{a'}}
\label{fig:exampleRulesAES}
\end{wrapfigure}

Examples motivating conditions 2 and 3 are in
Figs.~\ref{fig:AES:rule2}-\ref{fig:AES:rule3}, which present
situations where the merging of $a,a'$ does not preserve the behaviour
and hence should not be allowed. In the AES $\mathbb{A}_3$ of
Fig.~\ref{fig:AES:rule2}, $a \ac b''$ while neither $a' \ac b''$ nor
$b'' \ac a$, thus violating condition 2. In the AES $\mathbb{A}_3'$ of
Fig.~\ref{fig:AES:rule3}, $b'' \ac_\mu a$ while it is not the case
that $b'' \ac a'$, thus violating condition 3.

For an event $e_X$ resulting as the merging of a set of similar events
$X$, the events in $W(X)$ will be weak causes or consistent with
$e_X$. As a consequence, in order not to modify
the overall behaviour, all consistent subsets of $W(X)$ should match possible history of an event in $X$ already present in the original AES.
%
This is formalised by the definition of the combinable set of events.

\begin{definition}[combinable set of events]
  \label{def:combinableEvts}
  Let $\mathbb{A}=\langle E, \leq, \ac, \lambda \rangle$ be an AES. A
  set of events $X \subseteq E$ of equivalent events is
  \emph{combinable} if $\forall Y \subseteq W(X)$ \emph{consistent,}
  there exists $e \in X$ such that for all $e' \in Y$, $\neg (e \ac
  e')$ and there is $h_{e} \in \hist{e}$ satisfying $h^-_{e} \subseteq
  S(X) \cup \causes{Y}$.
\end{definition}

Armed with the above definitions, we can now formally introduce the folding of an AES.

\begin{definition}[folding of an AES]
  \label{def:foldingAES}
  Let $\mathbb{A}$ be an
  AES, $X$ be a set of combinable events. The \emph{folding} of
  $\mathbb{A}$ on $X$ is the AES 
  $\fold{\mathbb{A}}{X} = \langle \fold{E}{X}, \ <_{/X},\ac_{/X},\lambda_{/X}\rangle$
  where
  \begin{center}
    $
    \begin{array}{l>{\centering}p{8mm}l}
      \fold{E}{X} &=& 
      E \setminus X \cup \{e_X\},\\
      %
      \leq_X &=& 
      \leq_{|(E\setminus X)} \cup \{(e, e_X)\mid e \in S(X)\} \cup \{(e_X, e)\ |\exists e' \in X : e' \less e\} ,\\
      %
      \ac_X &=& 
      \ac_{|(E\setminus X)} 
      \cup \{(e', e_X) \mid \forall e \in X : e' \ac e \}
      \cup \{(e_X, e') \mid \forall e \in X : e \ac e'\}\\
      %
      \lambda_{/X} &=& \lambda,\lambda_{/X}(e_X) = \lambda(e) \text{ for an event } e \in X. \\
    \end{array}
    $
  \end{center}
\end{definition}

In words, the folding of $\mathbb{A}$ is obtained by replacing the set
$X$ of events with a single event $e_X$, with the same label as those
in $X$. The causes of $e_X$ are the common causes $S(X)$ of the events in $X$, and $e_X$ is a cause for all events caused by at least one event in $X$. The asymmetric conflicts for $e_X$ are exactly those of the events in $W(X)$.

It can be shown that $\fold{\mathbb{A}}{X}$ is indeed a properly
defined AES (the proof can be found in the Appendix).

\ifTechnicalReport
\else 
In order to show that the
folding operation preserves the behaviour, i.e., that the original and
folded AESs are hp-bisimilar, we rely on the notion of
AES-morphism~\cite{BaldanCM01}. Intuitively, an AES-morphism is a
mapping between AESs which shows how the target AES can simulate the
source AES. 
%
\fi

\begin{definition}[folding morphism]
  \label{def:mapping}
  Let $\mathbb{A}$ be an
  AES and let $X \subseteq E$ be combinable.  The \emph{folding
    map} $f: E\rightarrow \fold{E}{X}$ is defined as follows:
  \begin{center}
    $
    f(e) = 
    \left\{ 
      \begin{array}{cp{2mm}l}
	e_X & &\text{if } e \in X\\
	e && \text{otherwise}
      \end{array} 
    \right.
    $
  \end{center}
\end{definition}

It can be shown that the folding morphism is indeed an AES-morphism,
and as such it maps configurations of $\mathbb{A}$ into configurations
of $\fold{\mathbb{A}}{X}$.


Actually, according to the next lemma, the folding morphism have very
special properties, as it preserves and reflects asymmetric conflict in configurations.

\begin{lemma}
  Let $\mathbb{A}$ be an AES, and let $\fold{\mathbb{A}}{X} = \langle
  \fold{E}{X}, \leq_{/X},\ac_{/X}, \lambda_{/X} \rangle$ be the
  folding of $\mathbb{A}$ on the set of events $X$. Let $f :
  \mathbb{A} \rightarrow \fold{\mathbb{A}}{X}$ be the folding
  morphism. Then for any configuration $C_1 \in \conf{\mathbb{A}}$ it
  holds that $f(C_1) \in \conf{\fold{\mathbb{A}}{X}}$ and $(C_1,
  \ac_{C_1}^*) \approx (f(C_1), \ac_{f(C_1)}^*)$.
  \label{lemma:AESconfs}
\end{lemma}

The above result helps in proving that the folding morphism can be
seen as a hp-bisimilarity between $\mathbb{A}$ and
$\fold{\mathbb{A}}{X}$. Proofs can be found in the Appendix.

\begin{lemma}
  Let $\mathbb{A}$ be an AES,
  and let $\fold{\mathbb{A}}{X} = \langle \fold{E}{X}, \leq_{/X},\ac_{/X},
  \lambda_{/X} \rangle$ be the folding of $\mathbb{A}$ on the set of
  events $X$. Let $f : \mathbb{A} \rightarrow \fold{\mathbb{A}}{X}$ be the
  folding morphism. Then
  \begin{center}
    $R = \{ (C_1, f_{|C_1}, f(C_1)) \mid C_1 \in \conf{\mathbb{A}} \}$
  \end{center}
  is a hp-bisimulation.
  \label{lemma:AEShpbisim}
\end{lemma}

\begin{corollary}[folding does not change the behavior]
  The folding operation of AESs preserves hp-bisimilarity.
\end{corollary}

By iteratively applying folding to a given finite AES we can thus
obtain a minimal AES hp-bisimilar to the given one. Unfortunately,
this does not provide a canonical minimal representative of the
behaviour as there can be non-isomorphic minimal hp-bisimilar
AESs. For instance, consider the AES in
Figure~\ref{fig:exampleAES}(a). There exist two possible folded AESs,
presented side-by-side in Fig.~\ref{fig:AES:foldedAES}, which are
minimal in the sense that they cannot be further folded.

\begin{figure}[h!]
\centering
\subfloat[]{
  \includegraphics{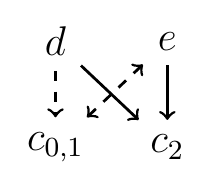}
}
\hspace{15mm}
\subfloat[]{
  \includegraphics{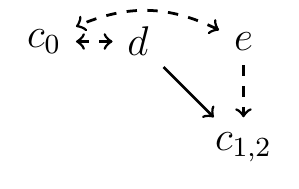}
}
\caption{Foldings for the AES in Fig.~\ref{fig:exampleAES}}
\label{fig:AES:foldedAES}
\end{figure}

\section{Behaviour preserving reduction of FESs}
\label{sec:reductionFES}

In this section we develop a behaviour preserving reduction technique
for flow event structures. As for AESs, the basic idea is that of folding 
events representing different instances of the same activity, although 
technically there are relevant differences. 

\subsection{Basics of flow event structures}

We start by recalling the formal definition of (labelled) flow event structures~\cite{BoudolC88}.

\begin{definition}[flow event structure]
  A (labelled) \emph{flow event structure} (FES) is a tuple
  $\mathbb{F} = \langle E, \#, \prec, \lambda \rangle$ where $E$ is a
  set of events, $\lambda : E
  \to \Lambda$ is a labelling function, and
  \begin{compactitem}
  \item $\prec ~ \subseteq E \times E$, the \emph{flow} relation, is
     irreflexive.
  \item $\# ~\subseteq E \times E$, the \emph{conflict} relation, is a
    symmetric relation,
  \end{compactitem}
\end{definition}

Note that the flow relation is not required to be transitive. The
$\prec$-predecessors of an event $e \in E$, are defined as $\pre{e} =
\{e' \mid e' \prec e\}$.  Similarly, for a set of events $X$ we write
$\pre{X} = \bigcup \{\pre{e} \mid e \in X\}$.

The flow predecessors of an event $e$, i.e., $\pre{e}$, can be seen as
a set of possible immediate causes for $e$.
Conflicts can exist in $\pre{e}$ and, in order to be
executed, $e$ needs to be preceded by a maximal and conflict free subset
of $\pre{e}$.

Formally, a \emph{configuration} of a FES $\mathbb{F} = \langle E, \#,
\prec,\lambda \rangle$ is a finite subset $C \subseteq E$
such that

\begin{compactenum}
\item $C$ is conflict free, 
\item $C$ has no flow cycles, i.e. $\prec_C^*$ is a partial order,
\item for all $e \in C$ and $e' \notin C$ s.t. $e' \prec e$, there exists an $e'' \in C$ 
such that $e' \#e'' \prec e$.\\
\end{compactenum}

We denote by $Conf(\mathbb{F})$ the set of configurations of
$\mathbb{F}$. The extension order, as for PESs, is simply subset
inclusion.

Since in FESs the flow relation is not transitive and the conflict relation
does not adhere to the principle of heredity: even though two events are not
in conflict they might not appear together in any
configuration, and an event could be not executable at all. 
More precisely,  define the semantic conflict $\#_s$ as $e \#_s
e'$ when for any configuration $C \in \conf{\mathbb{F}}$, it does not
hold that $\{e, e' \} \subseteq C$. Then clearly $\# \subseteq \#_s$ and in general the inclusion is strict.

In line with the authors of~\cite{BoudolC88}, hereafter we restrict to the
subclass of FES, where for which:
\begin{enumerate}
  \item semantic conflict $\#_s$ coincides with conflict $\#$
  (\emph{faithfullness}), 
  \item conflict is irreflexive (\emph{fullness}), hence all events are executable.
\end{enumerate}

Observe that FESs generalise PESs in that, clearly, every PES can be
seen as a special FES where the flow relation is transitive and the
$\prec$-predecessors of any event are conflict free.

%
	
	

\subsection{Reduction of FESs}

As in the case of AESs, we identify sets of events which can be seen
as instances of the same activity and which can be merged into a
single event. As mentioned before, the way in which FES generalises PES
is somehow orthogonal to that of AESs. As a consequence, at a technical
level the conditions which define combinable events are quite
different.

Consider, for instance, the example in
Fig.~\ref{fig:FES:original}. First, if we take events $c_0$ and $c_1$
and try to merge them into a single event $c_{0,1}$, there would be no
way of updating the dependency relations while keeping the behaviour
unchanged (the resulting dependency between $b$ and the merged event
$c_{0,1}$ would be an asymmetric conflict that cannot be represented
in FESs).  Instead, we can merge events $c_1$ and $c_2$ in
$\mathbb{F}_0$ into a single event $c_{1,2}$, thus obtaining the FES
in Fig.~\ref{fig:FES:folded}. In this case, the folding is possible
because the original events $c_1$ and $c_2$ are enabled by $\{ b \}$
and $\{d, e\}$, respectively, and since $b \# d$, $b \# e$, after the
merge the same situation is properly represented as an disjunctive
causality.


\begin{figure}
\centering
\subfloat[\label{fig:FES:original}$\mathbb{F}_0$] {
  \includegraphics{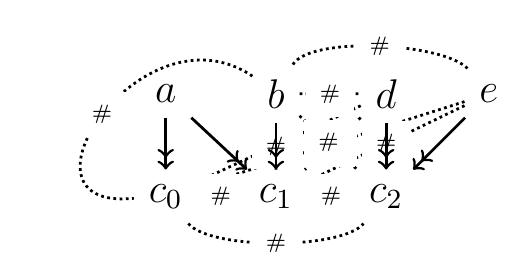}
}
\subfloat[\label{fig:FES:folded}$\mathbb{F}_1$]{
   \includegraphics{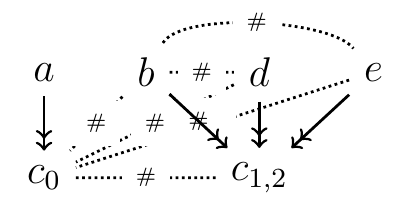}
}
\caption{Two sample FESs}
\label{fig:runningExampleFES}
\end{figure}

In order to define combinable events we need some further notation.
Given a set of events $Z$, we denote by $\mcons{Z}$ the set of
maximal and consistent (i.e., conflict free) subsets of $Z$. Given an
event $e \in E$, we write $\#(e)$ for the set of events in conflict
with $e$, i.e., $\#(e) = \{e' \mid e' \in E \land e' \# e\}$
Additionally, as for the case of AESs we need to single out conflicts
which are direct.

\begin{definition}[direct conflict]
\label{def:FES:directconf}
  Let $\mathbb{F}$ be a FES. The events $e,e' \in E$ are in
  \emph{direct conflict}, denoted as $e \dConflict e'$, if $e \# e'$
  and $\exists Y \in \mcons{\pre{e}}$ s.t. $Y \cap \#(e') =
  \emptyset$
  \noindent
\end{definition}

Intuitively, a conflict $e \# e'$ is direct when there is a way of
reaching a configuration where $e$ is enabled, without disabling
$e'$. Note that for FESs direct conflict is not symmetric. For
instance for $\mathbb{F}_4$ depicted in
Fig.~\ref{fig:FES:rule5:original}, we have $e \#_\mu a_1$ while it is
not the case that $a_1 \#_\mu e$. 

For a set $X \subseteq E$ and $e' \in E$ we write $X \# e'$ whenever
for all $e \in X$, we have $e \# e'$, $X \prec e'$ when there exists
$e \in X$ such that $e \prec e'$ and $e' \prec X$ when there exists $e
\in X$ such that $e' \prec e$.

We can now define the notion of combinable set of events for FESs.

\begin{definition}[combinable set of events]
  Let $\mathbb{F}$ be a FES. A set of events $X \subseteq E$ is
  combinable if for all $x,x' \in X$ and $e,e' \in E \setminus X$ the
  following holds
  \begin{compactenum}
  \item $\lambda(x) = \lambda(x')$ and $x \# x'$,
  \item $x \dConflict e \Rightarrow x' \# e$,		
  \item $x \prec e \Rightarrow x' \prec e ~\lor~ x' \# e$, 
  \item $e \prec x \Rightarrow \pre{x'} \neq \emptyset ~\land~ (e \prec x' ~\lor~ (\forall e' \prec x' \land e' \nin \pre{x}.\ e \# e'))$, 
  \item $x, e' \in \pre{e}\ \ \land\ \ x \# e'\land\ \ \neg (X \# e')\\
    ~\qquad \Rightarrow 
    \forall Y \in \mcons{\pre{e}} .\begin{array}{l}
        (x \in Y \Rightarrow \exists e'' \in Y \setminus \{x\}.\  
        e'' \# e') \land\\  
        (X \cap Y = \emptyset \Rightarrow \exists e'' \in Y.\  X \# e'')
  \end{array}
  $ 
    
  \end{compactenum}
  \label{def:combinableEvtsFES}
\end{definition}

Roughly speaking, condition 1 requires that the events in $X$ are
occurrences of the same activity (they have the same label and they
are in conflict).  Condition 2 requires that events in $X$ have
essentially the same conflicts. Conditions 3 and 4 state that
predecessors and successors are preserved among events in $X$ or they
can be turned into conflicts. 

The role of condition 4 is better explained by the following easy lemma.

\begin{lemma}
  \label{lemma:consistentSEfes}
  Let $\mathbb{F}$ be a FES and let $X\subseteq E$ be a combinable set
  of events. Then for any $Y \subseteq E$, $Y$ consistent it holds
  that $Y \subseteq \pre{X}$ iff there exists $x \in X$ such that $Y
  \subseteq \pre{x}$. Hence:
  \begin{center}
    $Y \in \mcons{\pre{X}}$ iff there exists $x \in X$ such that
    $Y \in \mcons{\pre{x}}$.
  \end{center}
\end{lemma}

\begin{proof}
  Concerning the first statement, let $Y \subseteq E$ be a consistent
  set of events. If $Y \subseteq \pre{X}$ then,
  Definition~\ref{def:combinableEvtsFES}(4), immediately implies that
  there exists $x \in X$ such that $Y \subseteq \pre{x}$. The converse
  implication just follows form the fact that $\pre{X} = \bigcup_{x
    \in X} \pre{x}$.

  The second statement is a trivial consequence of the first one.

\end{proof}

Finally, condition 5 takes into account the situation in which an
event $x' \in X$ is a potential cause of an event $e$, but there is
another $x \in X$ with different conflicts, say $x \# e'$, while
$\lnot (x' \# e')$. This is problematic, since after the merging the
unmatched conflict will be lost. The condition says that folding
can be still possible if the conflict $x' \# e'$ is not essential when
forming the maximal consistent sets of $\prec$-predecessors for $e$.
For example, the FES depicted in Fig.~\ref{fig:FES:rule5:couter}
illustrates a situation in which condition 5 fails.  Please note that
in Fig.~\ref{fig:FES:rule5:couter} events corresponding to those in
condition 5 have a subscript aimed at facilitating the
analysis. Merging $a_{x'}$ and $a_x$ would lead to the FES in
Figure~\ref{fig:FES:rule5:unfaithful}, which is not behaviourally
equivalent to $\mathbb{F}_2$. In particular, observe that $c$ is no
longer executable since $b \# e$, but it is not the case that $b \# a$
(hence such FES is not faithful). The conflict $b \# a$ could not be
imposed in the folded FES $\mathbb{F}_3$ otherwise a configuration
corresponding to $\{ d, a_x, b_{e'} \}$ in $\mathbb{F}_2$ would be
lost.
An allowed folding is shown in
Figures~\ref{fig:FES:rule5:original}-\ref{fig:FES:rule5:folded}: in
$\mathbb{F}_5$, after the execution of $e$ or $d$, it is possible to
have a maximal and consistent set of $\prec$-predecessors for the event
$c$, i.e., $\{a_{0,1}, f\}$ or $\{a_{0,1}, b\}$.

\begin{figure}
\centering
 \subfloat[\label{fig:FES:rule5:couter}$\mathbb{F}_2$] {
  \includegraphics{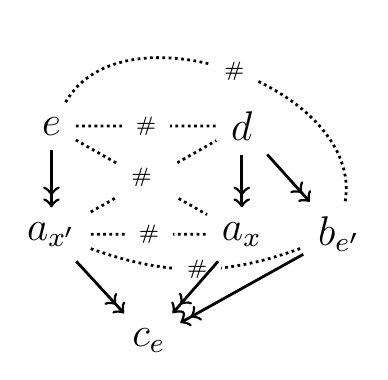}
}
\subfloat[\label{fig:FES:rule5:unfaithful}$\mathbb{F}_3$] {
   \includegraphics{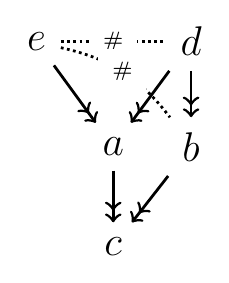}
}

\subfloat[\label{fig:FES:rule5:original}$\mathbb{F}_4$]{
   \includegraphics{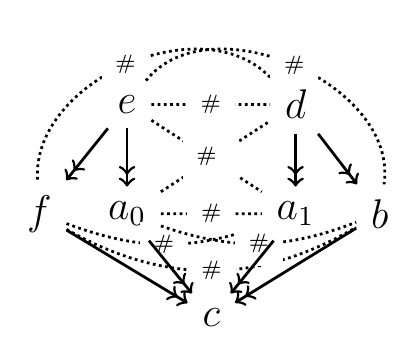}
 }
 \subfloat[\label{fig:FES:rule5:folded}$\mathbb{F}_5$]{
   \includegraphics{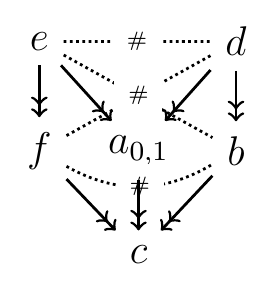}
 }
 \vspace{-2mm}
  \caption{Example FESs to illustrate Condition 5 in Definition~\ref{def:combinableEvtsFES}}
  \label{fig:FES:rule5}
  \end{figure}

  We can now formally define the folding of FESs as follows.

\begin{definition}[folding of FES]
  \label{def:foldingFES}
  Let $\mathbb{F}$ be a FES, $X$ be a set of combinable events. The
  \emph{folding} of $\mathbb{F}$ on $X$ is the FES
  $\fold{\mathbb{F}}{X} = \langle \fold{E}{X},
  \#_{/X},\prec_{/X},\lambda_{/X}\rangle$ where
  \begin{center}
    $
    \begin{array}{l>{\centering}p{8mm}l}
      \fold{E}{X} &=& 
      E \backslash X \cup \{e_X\},\\
      %
      \#_{/X} &=& 
      \#|_{(E\backslash X)} \cup \{(e, e_X) \mid e \# X \},\\
      %
      \prec_{/X} &=& 
      \prec|_{(E\backslash X)} \cup \{(e, e_X) \mid e \prec X \} \cup 
      \{(e_X, e') \mid X \prec e'\},\
    \end{array}
    $
  \end{center}
  and the labeling function defined by $\lambda_{/X}(e) = \lambda(e)$
  for $e \in E \setminus X$ while $\lambda_{/X}(e_X) = \lambda(e)$ for
  an event $e \in X$.
\end{definition}

The rest of the section is dedicated to show that the folding
operation on FESs preserves hp-bisimilarity. 

The idea, which underlies also AES folding, is that events which are
merged are occurrences of the same activity with different
histories. They can be merged if the histories are compatible and
after merging the possible histories remain the same. 
Since an event in a FES can occur after a maximal and consistent set of
$\prec$-predecessors (i.e., once all the conflict among its
predecessors has been resolved).  By Lemma~\ref{lemma:consistentSEfes}
above, after merging a set of combinable events this maximal subsets
of consistent events remains unchanged. This will be at the basis of
the proof that the merging does not alter the behaviour.

As in the case of AESs, we rely on the folding morphism.

\begin{definition}[folding morphism]
  \label{def:mappingFES}
  Let $\mathbb{F}$ be a FES and let $X \subseteq E$ be combinable.
  The \emph{folding function} $f: E\rightarrow \fold{E}{X}$ is defined
  as follows:
  \begin{center}
    $
    f(e) = 
    \left\{ 
      \begin{array}{cp{2mm}l}
	e_X & &\text{if } e \in X\\
	e && \text{otherwise}
      \end{array} 
    \right.
    $
  \end{center}
\end{definition}


We can prove that the folding morphism reflect conflicts, preserves the $\prec$-relation and it maps configurations
into configurations  (see Appendix for the detailed proof).
%
%
First, this can be used to show that the FES resulting from a folding is
faithful and full.

\begin{lemma}
  Let $\mathbb{F}$ be a FES, $\fold{\mathbb{F}}{X} = \langle
  \fold{E}{X}, \#_{/X}, \prec_{/X},\lambda_{/X}\rangle$ be the folded
  FES of $\mathbb{F}$ and $f : \mathbb{F} \rightarrow 
  \fold{\mathbb{F}}{X}$ be the folding morphism.  The FES
  $\fold{\mathbb{F}}{X}$ is
	\begin{inparaenum}[1)]
		\item \label{prop:faithfulProof} faithful, and
		\item \label{prop:conflictConsProof} full.
	\end{inparaenum}
	\label{lemma:PropertiesFESfaithfulConsistent}
\end{lemma}

Building on the previous technical results we can finally prove that 
the folding morphism $f$ can be seen as a hp-bisimulation.

\begin{lemma}
Let $\mathbb{F} = \langle E, \#, \prec, \lambda \rangle$ be a FES and  
$\fold{\mathbb{F}}{X} = \langle \fold{E}{X}, \#_{/X},\prec_{/X},\ \lambda_{/X}\rangle$ 
be a folded FES for an set of combinable events $X \subseteq E$. 
Let $f : \mathbb{F} \rightarrow \fold{\mathbb{F}}{X}$ be the folding morphism. 
Then
  \begin{center}
    $R = \{ (C_1, f_{|C_1}, f(C_1)) \mid C_1 \in \conf{\mathbb{F}} \}$
  \end{center}
  is a hp-bisimulation.

\label{lemma:FEShpbisim}
\end{lemma}

\begin{corollary}[folding does not change the behavior]
  The folding operation of FESs preserves hp-bisimilarity.
\end{corollary}

As for AESs the iterative application of folding to a given finite FES
allows one to minimise the given FES while preserving the
behaviour. Also in this case, there is no canonical representative,
i.e., there can be several minimal non-isomorphic FESs.



\section{Conclusion and future work}
\label{sec:conclusionFutureWork}

This paper presents reduction techniques, referred to as folding, for
AESs and FESs which allow one to reduce the number of events in an
event structure without changing the behaviour. The folding operation
merge sets of events that are intended to represent instances of the
same activity. The equivalence notion adopted is history preserving
bisimulation, a standard equivalence in the true concurrent
spectrum. Due to the different expressive power of AESs and FESs,
tailored folding techniques have been proposed for the two brands of
event structures.

It turned out that neither AESs nor FESs offer a canonical
representation of the behaviour of a process. More specifically, the
same process can have non-isomorphic and irreducible
foldings both in terms of AESs and FESs.
Therefore, a natural venue for future work is to investigate how
to characterise an ordering on foldings, leading to
a notion of minimal canonical AESs or FESs.

We noted that the conditions defining sets of combinable events
are orthogonal in both cases. In this respect, we envision a transformation
from AESs to FESs which would allow further folding at the price of
inserting unobservable events to simulate asymmetric conflict
on a FES. We contend that such a transformation would open the
possibility of taking advantage of the combined expressiveness
of AES and FES, possibly leading to more compact representations.
This is therefore another venue for future research.


Future work includes the assessment of performance (accuracy, efficiency) 
of the presented technique for process model differencing in real world 
process model collections. Naturally, it is planned to extend this work to 
cover cases with cycles. Finally, a promising avenue is the use of folding 
of FESs for approaching problems like process mining and elimination of 
duplicates in process models. An additional advantage of FESs is that they 
can easily transformed into a certain type of Petri nets, flow nets.


The minimisation of the behaviour of a process can be translated into
some kind of minimisation problem for automata or labelled transition
system. Most available techniques focus on interleaving behavioural
equivalences (like language or trace equivalence or various forms of
bisimilarity). We are not aware of approaches for the minimisation of
event structures or partially ordered models of computation. In some
cases, given a Petri net or an event structure a special transition
system can be extracted, on which minimisation is performed.
For instance in~\cite{MontanariP97} the authors propose an encoding of
safe Petri nets into a causal automata, in a way which preserves
hp-bisimilarity. The causal automata can be transformed into a
standard labelled transition system (LTS).  In this way, the LTS
representation can be used to check the equivalence between a pair of
processes or to find a minimal representation of the
behaviour. However, once a Petri net has been transformed into a causal
automaton, then it is not possible to obtain the Petri net
representation back, which can be of interest in some specific
applications.
In~\cite{vGlabbeek96}, the author uses a state transition diagram
referred to as process graph, for the representation of the behaviour
of a Petri net. Again, the transition diagram could be minimised with
some technique for LTSs with structured states, but not direct
approach is proposed.




\bibliographystyle{elsarticle-num} 
\bibliography{bibliography}

\ifTechnicalReport
\else
\newpage
\appendix
\section{Appendix}

\subsection{Proofs for \S~\ref{sec:reductionAES}:  Behaviour-Preserving Reduction of AESs}

Hereafter, to avoid the abuse of notation, given an AES $\mathbb{A} =\langle E, \leq, \ac, 
\lambda \rangle$ and a set of combinable events $X \subseteq E$, then the folding of 
$\mathbb{A}$ on the set of events $X$ is $\fold{\mathbb{A}}{X} = \langle
\fold{E}{X}, \leq_{/X},\ac_{/X}, \lambda_{/X} \rangle$.
  

\begin{restatable}{lemma}{LemmaPropertiesAESFoldingFunction}
  \label{lemma:PropertiesF}
  Let $\mathbb{A}$ be an AES,
  $X \subseteq E$ be combinable and let $f: E\rightarrow
  \fold{E}{X}$ be the folding morphism. Then for all $e \in E$, $x \in \fold{E}{X}$
  \begin{enumerate}
  \item if $x <_{/X}f(e)$ then there exists $e' \in E$ such that $e'
    < e$ and $f(e') =x$;
  \item if $f(e) \ac f(e')$ then $e \ac e'$;
  \item if $e \ac_\mu e'$ then $f(e) \ac_{/X} f(e')$ or $e \# e'$.

  \end{enumerate}
\end{restatable}
\begin{proof} 
  1. Let $x \in \fold{E}{X}$ and $e \in E$ be such that $x < f(e)$. 
  We distinguish various cases:
  \begin{compactitem}

  \item if $x = e_X$ then, by Definition~\ref{def:foldingAES}, there
    exists $e' \in X$ such that $e' < e$. Since $f(e') = e_X$ and
    $f(e) = e$, this is the desired conclusion.

  \item if $e \in X$ (and thus $f(e) = e_X$) then by
    Definition~\ref{def:foldingAES}, $x = e' \in S(X) \subseteq E -
    X$. Hence $e' < e''$ for all $e'' \in X$. In particular, hence $e'
    < e$, as desired.

  \item if none of the above apply, then $x = e' \in E$ and $f(e) = e$,
    hence the result trivially holds.
  \end{compactitem}
 
  \medskip

  2. Let $e, e' \in E$ and assume $f(e) \ac f(e')$. If $e \in X$ and
  thus $f(e) = e_X$ then, by Definition~\ref{def:foldingAES}, $e'' \ac
  e'$ for all $e'' \in X$. Thus in particular, $e \ac e'$ as
  desired. If instead, $e' \in X$ and thus $f(e') = e_X$ then, by
  Definition~\ref{def:foldingAES}, $e \ac e''$ for all $e'' \in
  X$. Thus in particular, $e \ac e'$ as desired.
  Finally, if $e, e' \not\in X$ then $f$ is the identity on $e, e'$,
  and thus the result trivially holds.

  \medskip

  3. Let $e, e' \in E$ and assume $e \ac_\mu e'$. We distinguish three cases:
  
  - If $e \in X$ then, by Definition~\ref{def:EquivalentEvts}(2),
  either $e' \ac e$ and thus $e \# e'$ and we are done, or for all
  $e'' \in X$ we have $e'' \ac e'$, hence $f(e) = e_X \ac e' = f(e')$, again
  as desired.

  - If $e' \in X$ then, by Definition~\ref{def:EquivalentEvts}(3), for
  all $e'' \in X$ we have $e \ac e''$ and thus $f(e) = e \ac e_X =
  f(e')$, as desired.

  - Otherwise, neither $e$ nor $e'$ are in $X$ and thus the thesis
  trivially follows.

\end{proof}

Note that the converse of (2) above, i.e., if $e \ac e'$ then $f(e) \ac f(e')$,
does not hold. For instance, consider the event structures in Figure~\ref{fig:aesFoldedAES}. If we merge the two $c$'s, we get that $a \ac c_1$ but it is not true that $f(a) \ac f(c_1)$.

\begin{figure}
\begin{center}
  \includegraphics{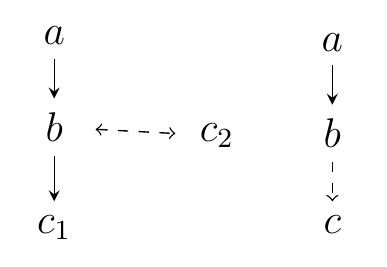}
\end{center}
\caption{AES and a folded structure}
\label{fig:aesFoldedAES}
\end{figure}

\begin{corollary}[reflection of $<$-chains]
  \label{cor:causal-reflect}
  With the notation of Lemma~\ref{lemma:PropertiesF}, take a chain
  $x_1 \leq x_2 \leq \ldots \leq x_K$ in $\fold{\mathbb{A}}{X}$. Then
  there is a chain $e_1 \leq e_2 \leq \ldots \leq e_K$ in
  $\mathbb{A}$, with $f(e_i) = x_i$ for $i \in \{ 1, \ldots, k\}$.
\end{corollary}

\begin{proof}
  It follows immediately by property (1) in
  Lemma~\ref{lemma:PropertiesF} and surjectivity of $f$.
\end{proof}

\begin{restatable}{lemma}{LemmaAEScompliance}
  \label{lemma:AES}
  Let $\mathbb{A}$ be an AES,
  let $X \subseteq E$ a combinable set. Then $\fold{\mathbb{A}}{X} = \langle
  \fold{E}{X}, \leq_{/X},\ac_{/X}, \lambda_{/X} \rangle$ is an AES.
\end{restatable}
\begin{proof}
  We first note that the transitivity of $\leq$ in $\fold{\mathbb{A}}{X}$
  (as defined in Definition~\ref{def:foldingAES}) follows immediately by
  transitivity of $\leq$ in $\mathbb{A}$.
  Similarly, asymmetric conflict is saturated in
  $\fold{\mathbb{A}}{X}$ because it was in $\mathbb{A}$. In fact, let
  $x, x', x'' \in \fold{E}{X}$ and assume that $x \ac x' < x''$. We
  prove that $x \ac x''$. We consider several cases. If $x = e_X$ then by
  Definition~\ref{def:foldingAES} for all $e \in X$ we have $e \ac x'
  < x''$ in $\mathbb{A}$, hence being $\mathbb{A}$ saturated, $e \ac
  x''$ and thus $x=e_X \ac x''$. If $x'' = e_X$ then by
  Definition~\ref{def:foldingAES} for all $e'' \in X$ we have $x \ac
  x' < e''$ in $\mathbb{A}$, hence being $\mathbb{A}$ saturated, $x
  \ac e''$ and thus $x \ac e_X = x''$.  If $x' = e_X$ then by
  Definition~\ref{def:foldingAES} there exists $e' \in X$ such that
  $e' < x''$. Moreover, $x \ac e'$ and thus $x \ac x''$ in
  $\mathbb{A}$ and therefore $x \ac x''$ in
  $\fold{\mathbb{A}}{X}$. Finally, if none of $x, x', x'' \in X$ then
  the thesis trivially follows.
  
  \medskip

  Let $f : E \to E_{/X}$ be the folding morphism.  We next observe that
  the defining properties of AESs hold.

  \begin{enumerate}
  
  \item  $\leq_{/X}$ is a well-founded partial order

    By Corollary~\ref{cor:causal-reflect}, causality chains are
    reflected, hence an infinite descending chain $x_1 > x_2 > x_3 >
    \ldots$ in $\fold{\mathbb{A}}{X}$, would be reflected in an infinite
    descending chain $e_1 > e_2 > e_3 > \ldots$ in $\mathbb{A}$.\\

  \item $\lfloor x \rfloor_{\fold{\mathbb{A}}{X}} 
    = \{ x' \in \fold{E}{X}\ |\ x' \leq_{/X} x\}$ 
    is finite for all $x\in \fold{E}{X}$

    This follows again, immediately, from
    Lemma~\ref{lemma:PropertiesF}(1) and surjectivity of $f$: an event
    with infinitely many causes would be reflected to an event with
    infinitely many causes in $\mathbb{A}$.\\

  \item $\ac_{\lfloor x \rfloor_{\fold{\mathbb{A}}{X}}}$ is acyclic for all
    $x \in \fold{E}{X}$
  
    Let $x \in \fold{E}{X}$ be an event and suppose that $\causes{x}$
    contains a cycle $x_1 \ac_{/X} x_2 \ac_{/X} \dots \ac_{/X}
    x_1$. By surjectivity of $f$ we can find $e \in E$ such that $x
    =f(e)$. By Lemma~\ref{lemma:PropertiesF}(1), there are events
    $e_1, \ldots, e_n \in \causes{e}$ such that $f(e_i) = x_i$ for any
    $i \in \{ 1, \ldots, n\}$. By point (2) of the same lemma, $e_1
    \ac e_2 \ac \dots \ac e_1$. This contradicts the property of
    $\ac_{\causes{e}} \in \mathbb{A}$ being acyclic for any event
    $e\in \mathbb{A}$.
\end{enumerate}
\end{proof}


We next recall the notion of AES-morphism from~\cite{BaldanCM01}, restricted to the case of total mappings between events which is of interest here. We will later use the fact, proved in the cited paper, that AESs morphisms preserve configurations.

\begin{definition}[AES-morphism]
  \label{def:ads:morphism}
  Let $\mathbb{A}_1$ and $\mathbb{A}_2$ be AESs. An AES-morphism $f :
  \mathbb{A}_1 \rightarrow \mathbb{A}_2$ is a function $f :
  E_1 \to E_2$ such that, for all $e, e' \in E_1$:

  \begin{compactenum}
  \item \label{morphism:one}  $\lfloor f(e) \rfloor \subseteq f(\lfloor e\rfloor)$;  
  \item \label{morphism:two} $f(e) \ac f(e') \Rightarrow e \ac e'$;
  \item \label{morphism:three} $(f(e) = f(e')) \land (e \neq e') \Rightarrow e \# e'$.

\end{compactenum}
\end{definition}

\begin{restatable}{lemma}{LemmaAESFoldingMorphism}
  \label{lemma:morphism}
  Let $\mathbb{A}$ be an AES, $X \subseteq E$ be a combinable set of
  events and let $\fold{\mathbb{A}}{X} = \langle \fold{E}{X},
  \leq_{/X},\ac_{/X}, \ \lambda_{/X} \rangle$ be the folded event
  structure. Then the folding morphism $f : E \to \fold{E}{X}$ is an
  AES-morphism.
\end{restatable}

\begin{proof}
  \begin{compactitem}
  \item Properties~\ref{morphism:one} and~\ref{morphism:two}. These
    follow directly from Lemma~\ref{lemma:PropertiesF} (1) and (2),
    respectively.
   
  \item Property \ref{morphism:three}. By Definition
    \ref{def:mapping}, for any pair of events $e,e' \in E$, $e\neq e'$, if  $f(e) = f(e')$ implies  $e,e'
    \in X$. Hence, by construction, $e\#e'$.
  \end{compactitem}
\end{proof}

\begin{replemma}{lemma:AESconfs}
  Let $\mathbb{A}$ be an AES, and let $\fold{\mathbb{A}}{X} = \langle
  \fold{E}{X}, \leq_{/X},\ac_{/X}, \lambda_{/X} \rangle$ be the
  folding of $\mathbb{A}$ on the set of events $X$. Let $f :
  \mathbb{A} \rightarrow \fold{\mathbb{A}}{X}$ be the folding
  morphism. Then for any configuration $C_1 \in \conf{\mathbb{A}}$ it
  holds that $f(C_1) \in \conf{\fold{\mathbb{A}}{X}}$ and $(C_1,
  \ac_{C_1}^*) \approx (f(C_1), \ac_{f(C_1)}^*)$.
\end{replemma}

\begin{proof}
  Let $C_1 \in \conf{\mathbb{A}}$ be a configuration. The fact that
  $f(C_1)$ is a configuration in $\conf{\fold{A}{X}}$ follows from the
  general properties of AESs morphisms (see~\cite{BaldanCM01}) and the
  fact that by Lemma~\ref{lemma:morphism} the folding morphism is an
  AES morphism.

  In order to prove that $(C_1, \ac_{C_1}^*) \approx (f(C_1),
  \ac_{f(C_1)}^*)$ it suffices to observe that for all $e, e' \in C_1$
  we have that
  \begin{center}
    $e \ac e'$ \qquad iff \qquad $f(e) \ac f(e')$
  \end{center}
  The fact that $f(e) \ac f(e')$ implies $e \ac e'$ has been already
  proved in Lemma~\ref{lemma:PropertiesF}(2).
  Vice versa, let $e \ac e'$. Then by
  Lemma~\ref{lemma:PropertiesF}(3), either $f(e) \ac f(e')$ or $e \#
  e'$. Since the latter cannot hold, because $e, e' \in C$ which is a
  configuration, necessarily $f(e) \ac f(e')$, as desired.
\end{proof}

\begin{replemma}{lemma:AEShpbisim}
  Let $\mathbb{A}$ be an AES,
  and let $\fold{\mathbb{A}}{X} = \langle \fold{E}{X}, \leq_{/X},\ac_{/X},
  \lambda_{/X} \rangle$ be the folding of $\mathbb{A}$ on the set of
  events $X$. Let $f : \mathbb{A} \rightarrow \fold{\mathbb{A}}{X}$ be the
  folding morphism. Then
  \begin{center}
    $R = \{ (C_1, f_{|C_1}, f(C_1)) \mid C_1 \in \conf{\mathbb{A}} \}$
  \end{center}
  is a hp-bisimulation.
\end{replemma} 
\begin{proof}
  First of all notice that for any $C_1 \in \conf{\mathbb{A}}$, if we
  let $C_2 = f(C_1)$, then by Lemma~\ref{lemma:PropertiesF}, $f_{|C_1}
  : (C_1, \ac^*) \to (C_2, \ac^*)$, is an isomorphism of pomsets.

  Moreover, in order to conclude, we next prove that
  \begin{enumerate}
  \item if there is $e \in E$ such that $C_1 \sqsubseteq C_1 \cup \{ e
    \} \in \conf{\mathbb{A}}$ then $C_2 \sqsubseteq C_2 \cup \{ f(e)
    \} \in \conf{\fold{\mathbb{A}}{X}}$.

  \item if there is $x \in \fold{E}{X}$ such that $C_2 \sqsubseteq C_2
    \cup \{ x \} \in \conf{\fold{\mathbb{A}}{X}}$ then there is $e \in
    E$ such that $f(e) = x$ and $C_1 \sqsubseteq C_1 \cup \{ e \} \in
    \conf{\fold{\mathbb{A}}{X}}$.
  \end{enumerate}

  \bigskip

  1. Note that $C_2 \cup \{ f(e) \} = f(C_1 \cup \{ e
  \})$ is a configuration by Lemma~\ref{lemma:AESconfs}. Moreover $C_2
  \sqsubseteq C_2 \cup \{ f(e) \}$, namely there is no $e' \in C_1$
  such that $f(e) \ac f(e')$, otherwise by
  Lemma~\ref{lemma:PropertiesF}(2) we would have $e \ac e'$,
  contradicting $C_1 \sqsubseteq C_1 \cup \{ e \}$.
  
  \bigskip

  2. Assume that $C_2 \sqsubseteq C_2 \cup \{ x \} \in
  \conf{\fold{\mathbb{A}}{X}}$ for some $x \in \fold{E}{X}$. We
  distinguish two cases.

  \medskip

  2.a) $x  =  e \in E \setminus X$

  Take the (unique) f-counterimage of $e$ of $x$, namely $f(e) = x$. A
  key observation is that
  \begin{quote}
    there is no $e' \in C_1$ such that $e \ac e'$. \qquad (\dag)
  \end{quote}
  In fact, we can show that given $e' \in C_1$ such that $e \ac e'$
  then there exists $e'' \in C_1$ such that $x = f(e) \ac f(e'')$,
  contradicting that $C_2 \sqsubseteq C_2 \cup \{ x \}$.
  To see this, assume first that $e \ac_\mu e'$. If $e' \not\in X$
  then clearly $f(e) \ac f(e')$. If $e' \in X$ then by
  Definition~\ref{def:EquivalentEvts}(3) $e \ac e'''$ for all $e'''
  \in X$ and thus also in this case, by
  Definition~\ref{def:foldingAES}, $f(e) = e \ac e_x = f(e')$. Hence
  we can take $e'' = e'$. If instead the asymmetric conflict is not
  direct, then there exists $e'''$ such that $e \ac_\mu e''' <
  e'$. Since $e' \in C_1$ by causal closure also $e''' \in C$ and thus
  the same argument of the previous case allows to conclude.

  Now we can easily prove that $C_1 \cup \{ e \} \in
  \conf{\mathbb{A}}$. To show that $\causes{e} \subseteq C_1$, take
  $e' < e$. Since $e \not\in X$, by Definition~\ref{def:foldingAES},
  we have $f(e') < f(e)$ and thus $f(e') \in f(C_1)$. Take $e'' \in
  C_1$ such that $f(e'') = f(e')$. Then by
  Lemma~\ref{lemma:morphism}(3), if $e' \neq e''$, then $e' \#
  e''$. Then we would have $e \# e''$, hence $e \ac e''$ violating
  (\dag) above. Hence it must be $e' = e'' \in C_1$, as desired.  The
  absence of cycles of asymmetric conflict in $C_1 \cup \{ e \}$
  follows immediately by the same property in $C_1$ and property (\dag) above.

  Similarly, $C_1 \sqsubseteq C_1 \cup \{ e \}$ is given directly by
  (\dag) above.

  \medskip

  2.b) $x = e_X$

  Consider the set
  \begin{center}
    $Y = \{ e' \mid f(e') \in C_2\ \land\ e' \in W(X) \}$
  \end{center}
  Clearly, $Y \subseteq W(X)$ and $Y$ consistent. Hence, by
  Definition~\ref{def:combinableEvts}, there exists $e \in X$ such
  that for all $e' \in Y$ $\neg (e \ac e')$ and for some $h_e \in
  \hist{e}$ it holds $h_e^- \subseteq S(X) \cup \causes{Y}$.
  
  As in the previous case we observe that
  \begin{quote}
    there is no $e' \in C_1$ such that $e \ac e'$. \qquad (\dag)
  \end{quote}
  In fact, given $e' \in C_1$ such that $e \ac e'$ then by
  Definition~\ref{def:EquivalentEvts}(2) either $e'' \ac e'$ for all
  $e'' \in X$ or there exists $e'' \in X$ such that $\neg (e'' \ac
  e')$ and $e \# e'$. In the first case, we would have $x = e_X \ac
  f(e')$, contradicting the fact that $C_2 \sqsubseteq C_2 \cup \{ x
  \}$. In the second case, from $e \# e'$ we have $e' \ac e$ and,
  additionally, there is $e'' \in X$ such that $\neg (e'' \ac e')$. We
  distinguish two subcases, depending on whether the asymmetric
  conflict $e' \ac e$ is direct or not. If $e' \ac_\mu e$ then $e' \in
  W(X)$. Therefore $e' \in Y$, contradicting the fact that we should
  have for all $e' \in Y$ $\neg (e \ac e')$.

  Now $f(e) = e_X = x$. Moreover $h_e^- \subseteq C_1$. In fact from
  $Y \subseteq C_1$ and the causal closure of $C_1$ we get $\causes{Y}
  \subseteq C_1$. Moreover if $e' \in S(X)$ then $e' < e''$ for any
  $e'' \in X$ and therefore $f(e') < e_X = x$. Hence $f(e') \in
  f(C_1)$, but since $e' \in E \setminus X$ is mapped identically by
  the folding morphism, this implies that $e' \in C_1$. Hence $S(X)
  \subseteq C_1$. Summing up, $h_e^- = S(X) \cup \causes{Y} \subseteq
  C_1$. From (\dag), as in (2.a) we can derive that $C_1 \cup \{ e \}$
  does not include cycles of asymmetric conflict and thus $C_1 \cup \{
  e \} \in \conf{\fold{\mathbb{A}}{X}}$.
  
  Moreover, $C_1 \sqsubseteq C_1 \cup \{ e \}$ follows immediately by (\dag).
\end{proof}



\subsection{Proofs for \S~\ref{sec:reductionFES}:  Behaviour preserving reduction of FESs}

We prove some properties of the folding morphism for FESs, which will
be used in proofs. We do not rely on the notion of morphism
in~\cite{Castellani97parallelproduct}, which would be too strong for
our needs (in particular, condition (iii) of \cite[Definition
4]{Castellani97parallelproduct} is not satisfied by our folding
morphism). 
In what follows, let $\mathbb{F}$ be a FES, $\fold{\mathbb{F}}{X} = \langle
\fold{E}{X}, \#_{/X}, \prec_{/X},\lambda_{/X}\rangle$ be the folding of a FES
$\mathbb{F} = \langle E, \#, \prec,\lambda \rangle$ on a set of combinable 
events $X$.

\begin{restatable}{lemma}{LemmaPropertiesFESfolding}
  Let $\mathbb{F}$ be a FES, $X\subseteq E$ be a combinable set of
  events and $\fold{\mathbb{F}}{X} = \langle \fold{E}{X}, \#_{/X},
  \prec_{/X},\lambda_{/X}\rangle$ its folding and let $f : E \to
  \fold{E}{X}$ be the folding morphism. Then for all $e,e' \in E$, $x \in
  E_{/X}$:
  \begin{compactenum}
  \item \label{lemma:propertiesFESM:one} 
    $f(e) \#_{/X} f(e') \Rightarrow e \# e'$


  \item \label{lemma:propertiesFESM:three} 
    $e \prec e' \Rightarrow f(e) \prec f(e')$
    
  \item \label{lemma:propertiesFESM:four} 
    $f(e) \prec f(e') \Rightarrow e \prec e'\ \lor\ e \# e'$

  \item \label{lemma:propertiesFESM:five} $f(e) = f(e') \Rightarrow e = e'\ \lor\ e \# e'$.


  \end{compactenum}
  \label{lemma:propertiesFESM}
\end{restatable}

\begin{proof}
  \begin{compactitem}
  \item Property~\ref{lemma:propertiesFESM:one}. 
    Let $e,e' \in E$ and assume $f(e) \# f(e')$. Notice that at least
    one between $e$ and $e'$ is not in $X$, otherwise we would have
    $f(e) = f(e')$. We distinguish various cases. If $e \in X$ and
    thus $f(e) = e_X$, then by definition of conflict in the folded
    FES (Definition~\ref{def:foldingFES}), since $f(e) = e_X \#
    f(e')$, it must be $e'' \# e'$ for all $e'' \in X$, and thus in
    particular $e \# e'$, as desired. The case in which $e' \in X$ is
    analogous, since conflict is symmetric.  Otherwise, if $e,e' \nin
    X$ the property trivially holds, since $f$ is the identity on
    $e,e'$.
    
    \medskip
    

    
    \medskip
    
  \item Property~\ref{lemma:propertiesFESM:three}. Let $e,e' \in E$ be
    such that $e \prec e' $. We distinguish the following cases:

    \begin{compactitem}
    \item $e \in X$.  
      By Definition~\ref{def:combinableEvtsFES}(1),
      $e' \nin X$ and, by Definition~\ref{def:foldingFES}, $e_X = f(e)
      \prec f(e') = e'$ as desired.

    \item $e' \in X$.
      As before, since $e' \in X$, then $e \nin X$ by
      Definition~\ref{def:combinableEvtsFES}(1).  Therefore, by
      construction, $e =f(e) \prec f(e')= e_X $.

    \item otherwise, if $e,e' \nin X$ then $f$ is the identity on $e,e'$ 
      and the result trivially holds.
    \end{compactitem}

  \item Property~\ref{lemma:propertiesFESM:four}. Let $e,e' \in E$ be
    such that $f(e) \prec f(e')$.
    Consider the following cases:
    \begin{compactitem}
    \item $e \in X$.  By Definition~\ref{def:combinableEvtsFES}(1),
      $e' \nin X$ and, by construction, there exists $e'' \in X$ such
      that $e'' \prec e'$. Then, either $e'' = e$ and thus $e \prec
      e'$, or, by Definition~\ref{def:combinableEvtsFES}(3), $e' \# e$
      as desired.
    \item $e' \in X$. As before, since $e' \in X$, then $e \nin X$ by Definition~\ref{def:combinableEvtsFES}(1) and, by construction, there exists $e'' \in X$ such
      that $e \prec e''$. Then, either $e'' = e'$ and thus $e \prec
      e'$, or, by Definition~\ref{def:combinableEvtsFES}(4), $e' \# e$
      as desired.

    \item otherwise, if $e,e' \nin X$ then $f$ is the identity on $e,e'$ 
      and hence $e \prec e'$.
    \end{compactitem}

  \item Property~\ref{lemma:propertiesFESM:five}.  Let $e, e' \in E$
    such that $f(e) = f(e')$, with $e \neq e'$. Since the events in
    $X$ are pairwise conflictual by
    Definition~\ref{def:combinableEvtsFES}(1), it is immediate to
    conclude that $e \# e'$.

  \end{compactitem}
\end{proof}



%
%
\begin{restatable}{lemma}{LemmaFESFoldingMorphism}
Let $\mathbb{F}$ be a FES, $X\subseteq E$ be a combinable 
set of events and  $f : \mathbb{F} \rightarrow \fold{\mathbb{F}}{X}$ be a folding 
morphism. For any configuration $C_0 \in Conf(\mathbb{F})$ then $f(C_0) = Conf(\fold{\mathbb{F}}{X})$ 
is a configuration in $\fold{\mathbb{F}}{X}$.
\label{lemma:configFES}
\end{restatable}

\begin{proof}
\begin{enumerate}
\item $f(C_0) $ is conflict free.

  This follows directly from
  Lemma~\ref{lemma:propertiesFESM}(\ref{lemma:propertiesFESM:one}). In
  fact, for $e_1, e_2 \in C_1$ if it were $f(e_1) \# f(e_2)$, then it
  would hold $e_1 \# e_2$.\\

\item $f(C_0) $ has no $\prec$-cycles.
  
  Observe that, inside configurations, by
  Lemma~\ref{lemma:propertiesFESM}(\ref{lemma:propertiesFESM:four}),
  the flow relation is reflected, namely for $e_1, e_2 \in C_1$, if
  $f(e_1) \prec f(e_2)$ then $e_1
  \prec e_2$ (since the case $e_1 \# e_2$ cannot apply). As a consequence, a $\prec$-cycle in $f(C_0)$ would be reflected in $C_0$.\\

\item $\{e_1' \mid e_1' \in f(E) ~\land~ e_1' \leq_{f(C_0)} e_1\}$ is
  finite for all $e_1 \in f(C_0) $. 

  This follows by the fact that the same property holds in $C_0$, since, 
  as observed above, $\prec$ is reflected inside configurations.\\
  
\item For all $e' \in f(C_0)$ and $e_1' \notin f(C_0) $ s.t. $e_1'
  \prec e'$, there exists $e_2' \in f(C_0) $ such that $e_1' \#e_2'
  \prec e'$.
  
  Let $e' \in f(C_0)$, $e_1' \nin f(C_0)$, such that $e_1' \prec
  e'$. Therefore, there are $e \in C_0$ such that $e' = f(e)$ and, by
  surjectivity of $f$, $e_1 \nin C_0$ such that $e_1' = f(e_1)$.
  
  By Lemma~\ref{lemma:propertiesFESM}(\ref{lemma:propertiesFESM:four})
  either $e_1 \prec e$ or $e_1 \# e$. In the last case, i.e., if $e_1
  \# e$ then necessarily by construction
  (Definition~\ref{def:foldingFES}), it must be that $e \in X$ is a
  folded event and there exists $e_3 \in X$ such that $e_1 \prec
  e_3$. Note that the conflict $e_1 \# e$ cannot be direct, otherwise,
  by Definition~\ref{def:combinableEvtsFES}(2), one should have also
  $e \# e_3$. Hence, since by definition of configuration, the set
  $\pre{e} \cap C_0 \in \mcons{\pre{e}}$, there must be $e_2 \in
  \pre{e} \cap C_0$ such that $e_1 \# e_2$. Hence $e_2 \in C_0$ and
  $e_2 \prec e$, therefore
  Lemma~\ref{lemma:propertiesFESM}(\ref{lemma:propertiesFESM:three}),
  $f(e_2) \prec f(e) = e'$. Moreover, since $e_1, e_2 \nin X$, we have
  $f(e_2) \# f(e_1) = e_1'$, as desired.

  Let us focus on the other case, in which $e_1 \prec e$.  Since
  $C_0$ is a configuration, there exists $e_2 \in C_0$ such that $e_2
  \prec e$ and $e_2 \# e_1$. By Lemma~\ref{lemma:propertiesFESM:four},
  $f(e_2) \prec f(e) = e'$. We distinguish various
  subcases:
  \begin{enumerate}
  \item $\{e_1, e_2\} \subseteq X$.
    This simply cannot happen as it would imply $f(e_1) = f(e_2) \in
    f(C_0)$, while we are assuming $f(e_1) \nin f(C_0)$.
    
  \item $e_1 \in X, e_2 \nin X$. 
    Let $Y \in \mcons{\pre{e}}$ be the set of maximal and consistent
    set of predecessors of $e$ in $C_0$. Obviously, $e_2 \in Y$ and, 
    by Lemma~\ref{lemma:propertiesFESM}, $\forall e_3 \in Y. f(e_3) 
    \prec f(e) = e'$ and $f(e_3) \in f(C_0)$. 
    Clearly, $\nexists e_4 \in Y \cap X. e_4 \in C_0$, otherwise $f(e_4) = f(e_1) 
    = e_1' \in f(C_0)$ and it would contradict the assumptions. Therefore, 
    $Y \cap X = \emptyset$ and, by Definition~\ref{def:combinableEvtsFES}(5),
    $\exists e_5 \in Y. e_5 \# X$. 
    In this case, by construction, $f(e_5) \# f(e_1) = e_1' = e_x$ and 
    since $f(e_5) \in f(C_0)$ then we obtain the desired result.
    
  \item $e_1 \nin X, e_2 \in X$.
    By Definition~\ref{def:combinableEvtsFES}(5), for all $Y \in
    \mcons{\pre{e}}$, with $e_2 \in Y$ there is $e_3 \in Y \setminus
    \{ e_2\}$ such that $e_3 \# e_1$. Since neither $e_1$ nor $e_3$
    are in $X$, this conflict is preserved by the folding morphism and
    thus $f(e_3) \# f(e_1) =e_1'$, as desired.

  \item $\{e_1, e_2\} \nsubseteq X$. Since $\{e_1, e_2\} \nsubseteq X$
    and $e_1 \# e_2$ then, by
    Lemma~\ref{lemma:propertiesFESM}(\ref{lemma:propertiesFESM:one}),
    $f(e_2) \# f(e_1)$, which contradicts the assumption. Hence also
    this case cannot happen.

  \end{enumerate}		
\end{enumerate}

\end{proof}

Recall that FESs are assumed to be faithful and full. We next prove that they remain so also after folding.

\begin{replemma}{lemma:PropertiesFESfaithfulConsistent}
  Let $\mathbb{F}$ be a FES, $\fold{\mathbb{F}}{X} = \langle
  \fold{E}{X}, \#_{/X}, \prec_{/X},\lambda_{/X}\rangle$ be the folded
  FES of $\mathbb{F}$ and $f : \mathbb{F} \rightarrow 
  \fold{\mathbb{F}}{X}$ be the folding morphism.  The FES
  $\fold{\mathbb{F}}{X}$ is
	\begin{inparaenum}[1)]
		\item \label{prop:faithfulProof} faithful, and
		\item \label{prop:conflictConsProof} full.
	\end{inparaenum}
\end{replemma}
\begin{proof} 
  
  \begin{compactenum}[-]
  \item Property \ref{prop:faithfulProof}). Faithfulness
    
    Let $x,x' \in \fold{E}{X} : \lnot(x \# x')$ be a pair of events in
    $\fold{\mathbb{F}}{X}$.  We need to prove that there exists a
    configuration $C_1 \in \conf{\fold{E}{X}}$ such that $\{ x,x' \}
    \subseteq C_1$. 

    Take $e, e' \in E$ such that $f(e) = x$ and $f(e') = x'$ (they
    exist since $f$ is surjective). If $\lnot(e \# e')$ then by
    faithfullness of $\mathbb{F}$ there exists $C_0 \in
    Conf(\mathbb{F})$ such that $\{e,e'\} \subseteq C_0$. By
    Lemma~\ref{lemma:configFES}, $f(C_0) \in
    Conf(\fold{\mathbb{F}}{X})$ is the desired configuration, since
    $\{x, x'\} = \{f(e), f(e')\} \subseteq f(C_0)$.
    
    If, instead $e \# e'$, it means that one of the two events is in
    $X$. Assume without loss of generality that $e \in X$ and $e'
    \not\in X$. The fact that $\lnot (f(e) \# f(e'))$ means that there
    is $e'' \in X$ such that $\lnot (e'' \# e')$. Therefore, again by
    fullness there exists $C_0 \in Conf(\mathbb{F})$ such that
    $\{e'',e'\} \subseteq C_0$ and we conclude as above. In fact,
    $f(e'') =f(e) =c$, hence $\{x, x'\} = \{f(e), f(e')\} \subseteq
    f(C_0)$, which is a configuration by
    Lemma~\ref{lemma:configFES}.     \\
				
  \item Property \ref{prop:conflictConsProof}). Fullness
    
    By
    Lemma~\ref{lemma:propertiesFESM}(\ref{lemma:propertiesFESM:one})
    and surjectivity of $f$, a self-conflicting (inconsistent) event in
    $\mathbb{F}_{/X}$ would be reflected in
    $\mathbb{F}$. More precisely, let $x \in \fold{\mathbb{F}}{X}$
    such that $x \# x$. Then take $e \in \mathbb{F}$ such that $f(e)
    =x$. We have $f(e) \# f(e)$ and thus, by
    Lemma~\ref{lemma:propertiesFESM}(\ref{lemma:propertiesFESM:one}),
    $e\# e$, contradicting the fullness of $\mathbb{F}$.

  \end{compactenum}

\end{proof}

%
%
\begin{replemma}{lemma:FEShpbisim}
Let $\mathbb{F} = \langle E, \#, \prec, \lambda \rangle$ be a FES and  
$\fold{\mathbb{F}}{X} = \langle \fold{E}{X}, \#_{/X},\prec_{/X},\ \lambda_{/X}\rangle$ 
be a folded FES for an set of combinable events $X \subseteq E$. 
Let $f : \mathbb{F} \rightarrow \fold{\mathbb{F}}{X}$ be the folding morphism. 
Then
  \begin{center}
    $R = \{ (C_1, f_{|C_1}, f(C_1)) \mid C_1 \in \conf{\mathbb{F}} \}$
  \end{center}
  is a hp-bisimulation.
\end{replemma}

\begin{proof}
  Given a configuration $C_1 \in \conf{\mathbb{F}}$, let $C_2 =
  f(C_1)$. Observe that $f: (C_1, \prec^*) \approx (C_2, \prec^*)$ is
  an isomorphims of pomsets. This follows immediately by items
  (\ref{lemma:propertiesFESM:three}) and
  (\ref{lemma:propertiesFESM:four}) of
  Lemma~\ref{lemma:propertiesFESM}.
  In order to show that $R$ is a
  hp-bisimilarity it remains to prove that
 \begin{enumerate}
 \item if there is $e \in E$ such that $C_1 \cup \{ e \} \in
   \conf{\mathbb{F}}$ then $C_2 \cup \{ f(e) \} \in
   \conf{\fold{\mathbb{F}}{X}}$.
   
 \item if there is $x \in \fold{E}{X}$ such that $C_2 \cup \{ x \} \in
   \conf{\mathbb{F}}$ then there is $e \in E$ such that $f(e) = x$,
   $C_1 \cup \{ e \} \in \conf{\mathbb{F}}$.
 \end{enumerate}
 
 \bigskip
 
 In the following, the subscript $/X$ in the relations of the folded
 FES are omitted for making the notation lighter.  
 
 \begin{enumerate}
 \item The fact that if $C_1 \cup \{ e \} \in \conf{\mathbb{F}}$ then
   $C_2 \cup \{ f(e) \} \in \conf{\fold{\mathbb{F}}{X}}$ follows
   immediately by Lemma~\ref{lemma:configFES}.

 \item Let $x \in \fold{E}{X}$ be such that $C_2 \cup \{ x \} \in
   \conf{\fold{\mathbb{F}}{X}}$. Thus, it is necessary to show that 
   there is an event $e \in E$ such that $f(e) = x$, $C_1 \cup \{ e \} \in
   \conf{\mathbb{F}}$ and $C_1 \cup \{ e \} \approx C_2 \cup \{ x \}$.
   
   Let $Y_2 = \pre{x} \cap C_2$ be the set of $\prec$-predecessors of
   $x$ in $C_2$. By definition of configuration in FESs we know that
   $Y_2 \in \mcons{\pre{x}}$.

   We distinguish two cases:

   \begin{enumerate}
   \item $x = e_X$.
     
     In this case events in $\pre{x}$ are left unchanged by the
     folding and hence if we let $Y_1 = Y_2$ we have that $Y_1
     \subseteq C_1$, $f(Y_1) = Y_2$ and $Y_1$ is consistent.  By
     definition of the folding $Y_1 \subseteq \pre{X}$ and thus by
     Lemma~\ref{lemma:consistentSEfes}, there is an event $e' \in X$,
     s.t.  $Y_1 \in \mcons{\pre{e'}}$. Since $Y_1 \subseteq C_1$, we
     deduce that $C_1 \cup \{e'\} \in Conf(\mathbb{F})$, and it holds
     $f(C_1 \cup \{e'\}) = C_2 \cup \{e_X\}$, as desired, since $f(e') = e_X$.
     
     \medskip

   \item $x \neq e_X$.
     In this case the event $x = e \in E \setminus X$ is mapped
     identically by the folding morphism $f$. We just need to show
     that $C_1 \cup \{ e \}$ is a configuration. Let $Y_1 = \{ e' \in
     C_1 \mid f(e') \in Y_2 \}$. 

     We have that $Y_1 \subseteq \pre{e}$. In fact, for any $e' \in
     Y_1$, since $f(e') \prec x$, by
     Lemma~\ref{lemma:propertiesFESM}(\ref{lemma:propertiesFESM:four})
     we know that $e' \prec e$ or $e' \# e$. The second case cannot
     happen, 
     since $\lnot f(e) \# f(e')$, by Definition~\ref{def:foldingFES} there
     is $e'' \in X$ such that $\lnot e \# e''$. Then by
     Definition~\ref{def:combinableEvtsFES}(2), the conflict $e' \# e$
     is not direct. Therefore, since $\pre{e'} \cap C_1 \in
     \mcons{\pre{e'}}$, by definition of direct conflict, there is
     $e''' \in \pre{e'} \cap C_1$ such that $e''' \# e$. Since
     $e'''\not\in X$, this conflict is preserved by the conflict
     morphims and we get that $f(e''') \# f(e)$, which is absurd as
     $f(e), f(e''') \in f(C_1) \cup \{ x \}$ which is a configuration
     by hypothesis.

     The set $Y_1$ is clearly consistent, as it is included in
     $C_1$. It is also maximal, i.e., $Y_1 \in
     \mcons{\pre{e}}$. In fact if it were not maximal, there would
     be $e'' \in \pre{e} \setminus Y_1$ such that $Y_1 \cup \{ e''\}$
     is consistent. But then, since the folding morphism preserves
     configurations and thus consistent sets, $f(Y_1 \cup \{ e''\})$
     would be consistent and strictly larger then $Y_2$.

     Since $Y_1 \in \mcons{\pre{e}}$, we conclude that $Y_1 \cup
     \{e\}$ is a configuration, as desired.

   \end{enumerate}
 \end{enumerate} 
 
\end{proof}





\fi

\end{document}